\documentclass{sig-alternate}

\usepackage{verbatim}
\usepackage{amsmath}
\usepackage{amssymb}
\usepackage{url}
\usepackage{array}
\usepackage{booktabs}
\usepackage{afterpage}
\usepackage[dvipsnames]{color}
\usepackage{mfirstuc}
\usepackage{subfigure}
\usepackage{times} 
\makeatletter
\newif\if@restonecol
\makeatother

\usepackage[lined,algonl,boxed]{algorithm2e}
\usepackage{epsfig}
\usepackage{multirow}
\usepackage{rotating}
\usepackage{enumitem}

\newcommand{\omt}[1]{}

\title{Graph Cluster Randomization: \\ Network Exposure to Multiple Universes}

\numberofauthors{1}
\author{
    \alignauthor 
    \vspace*{-0.2in} 
        Johan Ugander \hspace{1.0cm}
        Brian Karrer \hspace{1.1cm}
        Lars Backstrom \hspace{1.0cm}
        Jon Kleinberg\hspace*{0.25cm} \\
    \affaddr{
        Cornell University\hspace{1.5cm}
        Facebook \hspace{2.2cm}
        Facebook \hspace{1.7cm}
        Cornell University
        \hspace*{0.05cm}
        } \\
    \email{ {
      \hspace*{0.1cm}
      jhu5@cornell.edu \hspace{1.5cm}
      \{briankarrer,lars\}@fb.com \hspace{0.9cm}
      kleinber@cs.cornell.edu}}
}

\date{}

\newtheorem{theorem}{Theorem}[section]

\newtheorem{definition}[theorem]{Definition}
\newtheorem{proposition}[theorem]{Proposition}

\newcommand{\Cov}{{\text{Cov}}}
\newcommand{\Var}{{\text{Var}}}

\begin{document}

\maketitle

\begin{abstract}
A/B testing is a standard approach for evaluating the effect of 
online experiments; the goal is to estimate the `average treatment effect'
of a new feature or condition by exposing a sample of the overall
population to it.
A drawback with A/B testing is that it is poorly suited for
experiments involving social interference, when the treatment of
individuals spills over to neighboring individuals along an underlying
social network. In this work, we propose a novel methodology using
graph clustering to analyze average treatment effects under
social interference. To begin, we characterize graph-theoretic
conditions under which individuals can be considered to be `network
exposed' to an experiment.  We then show how graph cluster
randomization admits an efficient exact algorithm to compute the
probabilities for each vertex being network exposed under
several of these exposure conditions. Using these probabilities as
inverse weights, a Horvitz-Thompson estimator can then provide an
effect estimate that is unbiased, provided that the exposure model has
been properly specified.

Given an estimator that is unbiased, we focus on minimizing 
the variance.
First, we develop simple sufficient conditions for the variance 
of the estimator to be asymptotically small in $n$, the size of the graph. 
However, for general randomization schemes, 
this variance can be {\it lower bounded} by an {\it exponential} function of the degrees of a graph. 
In contrast, we show that if a graph satisfies a {\it restricted-growth condition}
on the growth rate of neighborhoods, then there
exists a natural clustering algorithm, based on vertex
neighborhoods,
for which the variance of the estimator can be {\it upper bounded} 
by a {\it linear} function of the degrees. 
Thus we show that proper cluster randomization can lead to 
exponentially lower estimator variance when experimentally 
measuring average treatment effects under interference.
\end{abstract}

\vspace{-0.15cm}
\section{Introduction}

\newcommand{\xhdr}[1]{\paragraph*{\bf {#1}}}

Social products and services -- from fax machines and cell phones to
online social networks -- inherently exhibit `network effects' with
regard to their value to users. The value of these products to a user
is inherently non-local, since it typically grows as members of the
user's social neighborhood use the product as well.
Yet randomized experiments
(or `A/B tests'), the standard machinery of testing frameworks
including the Rubin causal model
\cite{Rubin1974}, critically assume what is known as the `stable unit
treatment value assumption' (SUTVA), that each individual's response
is affected only by their own treatment and not by the treatment of
any other individual.
Addressing this tension between the formalism of A/B testing 
and the non-local effects of network interaction has emerged
as a key open question in the analysis of on-line behavior
and the design of network experiments \cite{fienberg2012}.

Under ordinary randomized trials where the stable unit treatment
value assumption is a reasonable approximation --- for example when a
search engine A/B tests the effect of their color scheme upon the
visitation time of their users --- the population is divided into two 
groups: those in the `treatment' group who see the new color scheme A and 
those in the control group who see the default color scheme B. 
Assuming there are negligible interference effects between users,
each individual in the treated group responds just as he or she would if the entire population were
treated, and each individual in the control group responds just as he or she would
if the entire population were in control. In this manner, we
can imagine that we are observing results from samples of two
distinct `parallel universes' at the same time --- `Universe A' in
which color scheme A is used for everyone, and 
`Universe B' in which color scheme B is used for everyone ---
and we can make inferences about the properties of user behavior
in each of these universes.

This tractable structure changes dramatically when the behavior of one
user $i$ can have a non-trivial effect on the behavior of another user $j$ ---
as is the case when the feature or product being tested has any kind of
social component. 
Now, if $i$ is placed in Universe A and $j$ is placed in Universe B,
then our analysis of $i$'s behavior in A is contaminated by 
properties of $j$'s behavior in B, and vice versa; we no longer have
two parallel universes.

\xhdr{Average Treatment and Network Exposure}
Our goal is to develop techniques for analyzing the
average effect of a treatment on a population when such interaction 
is present.
As our basic scenario, we imagine testing a service by providing
it to a subset of an underlying population; the service has a
`social' component in that
$i$'s reaction to the service depends on whether a neighbor $j$
in the social network also has the service. We say that an
individual is in the {\em treatment group} if the individual is provided 
with the service for the test, and in the {\em control group} otherwise.
There is an underlying numerical response variable of interest (for example,
the user's time-on-site in each condition), and we want to estimate
the average of this response in both the universe where everyone
has the service, and the universe where no one has the service,
despite the fact that --- since the population is divided
between treatment and control --- 
we don't have direct access to either universe.

We express this question using a formalism introduced by 
Aronow and Samii for causal inference without this
stable unit treatment value assumption \cite{AS2012}, 
with strong similarities to similar formalism introduce by Manski \cite{Manski2013},
and adapt it to the problem of interference on social networks.
Let $\vec{z} \in \{0,1\}^n$ be the treatment assignment vector, 
where $z_i = 1$ means that user $i$ is in the 
treatment group and $z_i = 0$ means the user is in the control.  
Let $Y_i(\vec{z}) \in \mathbb{R}$ be the potential outcome of 
user $i$ under the treatment assignment vector $\vec{z}$.  
The fundamental quantity we are interested in is the average treatment effect, $\tau$, between the two diametrically opposite 
universes $\vec z = \vec 1$ and $\vec z' = \vec 0$,
\begin{eqnarray}
\tau(\vec z = \vec 1, \vec z' = \vec 0) = \frac{1}{n}{\sum_{i=1}^n \left[ Y_i(\vec{z} = \vec 1 )-Y_i(\vec{z}' = \vec 0)\right]}.
\label{eq:avg-effect}
\end{eqnarray}
This formulation contains the core problem discussed in informal terms
above: unlike ordinary A/B testing, no two users can ever truly be 
in opposing universes at the same time.

A key notion that we introduce for evaluating (\ref{eq:avg-effect})
is the notion of {\em network exposure}.  We say that
$i$ is `network exposed' to the treatment under a particular 
assignment $\vec{z}'$ if $i$'s response under $\vec{z}'$ is the same
as $i$'s response in the assignment $\vec{1}$, where everyone receives
the treatment.\footnote{We also discuss adaptations to the case where the 
responses in these two cases differ only by a small parameter $\varepsilon$.}
We define network exposure to the control condition analogously.

With this definition in place, we can investigate several possible
conditions that constitute network exposure.
For example, one basic condition would be to say that $i$ is
network exposed to the treatment if $i$ and all of $i$'s neighbors are treated.
Another would be to fix a fraction $q > 0$ and say that $i$
is network exposed if $i$ and at least a $q$ fraction of $i$'s
neighbors are treated.
The definition of network exposure is fundamentally a modeling decision
by the experimenter, and in this work we introduce several families of
exposure conditions, each specifying the sets of
assignment vectors in which a user is assumed to be `network exposed' to the
treatment and control universes, providing several characterizations
of the continuum between the two universes. Choosing network exposure
conditions is crucial because they specify when we can observe the
potential outcome of a user as if they were in the treatment or
control universe, without actually placing all users into the
treatment or control universe.

\xhdr{Graph Cluster Randomization}
Following the formulation of network exposure,
a second key notion that we introduce
is a generic graph randomization scheme based on graph
clustering, which we call {\it graph cluster randomization}.
At a high level,  graph cluster randomization is a technique in 
which the graph is partitioned into a set of {\em clusters},
and then randomization between treatment and control is performed
at the cluster level.  
The probability that a vertex is network exposed to treatment or control
will then typically involve a graph-theoretic question about the
intersection of the set of clusters with the 
local graph structure near the vertex.
We show how it is
possible to precisely determine the non-uniform probabilities of
entering network exposure conditions under such randomization. Using
inverse probability weighting \cite{HT1952}, we are then able to
derive an unbiased estimator of the average treatment effect $\tau$
under any network exposure for which we can explicitly compute
probabilities.

We motivate the power of graph cluster randomization by furnishing
conditions under which graph cluster randomization will produce
an estimator with asymptotically small variance. 
First, we observe that if the graph has bounded degree 
and the sizes of all the clusters remain bounded independent of 
the number of vertices $n$,
then the estimator variance is $O(1/n)$, a simple but illustrative
sufficient condition for smallness. 
The key challenge is the dependence on the degrees --- in general,
a collection of bounded-size clusters can produce a variance that
grows exponentially in the vertex degrees.
More precisely, when performing graph cluster randomization with 
single-vertex clusters, the variance of the estimator admits a {\it
lower bound} that depends {\it exponentially} on the degrees. 
This raises the important algorithmic question of how to choose 
the clustering: bounded-size
clusters provide asymptotically small variance in the number of
vertices $n$, but if the clusters are not chosen carefully 
then we get an exponential dependence on the vertex degrees
which could cause the variance to be very large in practice.

\xhdr{Cluster Randomization in Restricted-Growth Graphs}
We identify an important class of graphs, which we call
{\em restricted-growth graphs}, 
on which a non-trivial clustering algorithm 
admits an {\it upper bound} on the
estimator variance that is {\it linear} in the degrees of the graph. 
The restricted-growth condition that we introduce for graphs is
an expansion of the bounded-growth condition previously
introduced for studying nearest-neighbor algorithms in metric spaces
\cite{karger2002}, designed to include low-diameter
graphs in which neighborhoods can grow exponentially.
Formally, let $B_{r}(v)$ be the set of vertices within $r$
hops of a vertex $v$; our restricted-growth condition says
that there exists a constant $\kappa$, independent of the degrees of
the graph, such that for all vertices $v$ and all $r>0$, we have
$|B_{r+1}(v)| \leq \kappa |B_r(v)|$. 
Note the comparison to the standard bounded-growth definition, 
which requires $|B_{2r}(v)| \leq \kappa |B_r(v)|$, a much stronger
condition and not necessary for our results to hold.

For restricted-growth graphs, we provide a clustering algorithm
for which the estimator variance grows only linearly in the degree.
The challenge is that the variance can grow exponentially with
the number of clusters that intersect a vertex's neighborhood;
our approach is to form clusters from balls of fixed radius grown 
around a set of well-separated vertices.
The restricted growth condition prevents balls from packing too
closely around any one vertex, thus preventing vertex neighborhoods from
meeting too many clusters.
We note that for the special case of restricted-growth graphs that
come with a uniform-density embedding in Euclidean space, one can use
the locations of vertices in the embedding to carve up the space into
clusters directly; the point, as in work on the nearest-neighbor problem
\cite{karger2002}, is to control this carving-up at a graph-theoretic
level rather than a geometric one, and this is what our technique does.

Our class of restricted-growth graphs provides an attractive
model for certain types of real-world graphs. 
Restricted-growth graphs include graphs for which there exists an
embedding of the vertices with approximately uniform density
in a Euclidean space of bounded dimension, such as lattices
or random geometric graphs, where edges connect neighbors within some
maximal metric distance. 

\xhdr{Summary}
Our work thus occupies a mediating perch between recent work from the
statistical literature on causal inference under interference
\cite{AKTR2012,AS2012,TV2012}, as well as recent work from the
computer science literature on network bucket testing
\cite{BK2011,KLS2012}. Our contribution extends upon the ordinary
inference literature by developing exposure models and randomization
schemes particularly suited for experiments on large social graphs,
also showing how previous approaches are intractable. 
Meanwhile, we show that reducing estimator variance involves 
non-trivial graph-theoretic considerations, and we introduce 
a clustering algorithm that improves exponentially on baseline
randomization schemes. Our
contribution also connects to existing work on network bucket testing
by contributing an exposure framework for the full graph and a
randomization scheme that is capable of considering multiple exposure
conditions at once, a necessity for true concurrent causal
experimentation. 

In Section 2 we describe our models of network exposure. In Section 3 we present our graph cluster randomization scheme, an algorithm for efficiently computing exposure probabilities, and an unbiased estimator of average treatment effects under graph cluster randomization. In Section 4 we introduce restricted-growth graphs, and show how the estimator has a variance that is linearly bounded in degree for such graphs. Section 5 concludes. 
 
\section{Network exposure models}

For A/B randomized experiments, the {\it treatment condition} of an
individual decides whether or not they are subject to an intervention.
This typically takes two values: `treatment' or `control'. In most
randomized experiments, the experimenter has explicit control over how
to randomize the treatment conditions, and generally individuals are
assigned independently. Meanwhile, the {\it exposure condition} of an
individual determines how they experience the intervention in full
conjunction with how the world experiences the intervention.  Without
the stable unit treatment value assumption, at worst each of the $2^n$
possible values of $\vec z$ define a distinct exposure condition for each
user. Aronow and Samii call this ``arbitrary exposure'' \cite{AS2012},
and there would be no tractable way to analyze experiments under
arbitrary exposure.

Consider the potential outcomes for user $i$. In the ``arbitrary exposure" case,
$Y_i(\vec{z})$ is completely different for every possible $\vec{z}$.  
This means that we will never be able to observe $Y_i(\vec{z})$ 
for either
$\vec{z}=\vec{1}$ or $\vec{z}=\vec{0}$ without putting all users into the 
treatment or control universes.  
Thus, to make progress on estimating the average treatment effect under
any other conditions, we require further assumptions.
We do this here by assuming that multiple treatment vectors $\vec{z}$
can map to the same potential outcomes: essentially, as long as
treatment vectors $\vec{z}$ and $\vec{z}'$ are ``similar enough''
from the perspective of a vertex $i$, in a sense to be made precise below,
then $i$ will have the same response under $\vec{z}$ and $\vec{z}'$.

Specifically, 
let
$\sigma^x_i$ be the set of all assignment vectors $\vec{z}$
for which $i$ experiences outcome $x$.
We refer to $\sigma^x_i$ as an {\em exposure condition} for $i$;
essentially, $\sigma^x_i$ consists of a set of assignment vectors
that are ``indistinguishble'' from $i$'s point of view, in that
their effects on $i$ are the same.
Our interest is in the particular exposure conditions
$\sigma^1_i$ and $\sigma^0_i$, which we define to be the sets that contain $\vec{z}= \vec 1$ 
and $\vec{z}= \vec 0$ respectively. 
In this way, we are assuming that for all $\vec z_1 \in \sigma^1_i$, 
we have $Y_i(\vec z = \vec z_1) = Y_i(\vec{z}=\vec 1)$, and for all 
$
\vec z_0 \in \sigma^0_i$, we have $Y_i(\vec z = \vec z_0) = Y_i(\vec{z}=\vec 0)
$.\footnote{
If this strikes the reader as too restrictive a definition of ``exposure condition'', consider instead partitioning the space of potential outcomes (rather that partitioning the space of assignment vectors) using small $\epsilon$-sized bins, and define the ``exposure conditions'' as all assignment vectors that produce a potential outcome in that $\epsilon$ bin. 
In cases where no other potential outcomes correspond to the outcomes for $\vec z = \vec 0 $ or $\vec z = \vec 1$, it may be more appropriate to manage bias using $\epsilon$ distances on potential outcomes this way.
}
Note that it is possible that $\vec{z}=\vec 1$ and $\vec{z}=\vec 0$ belong to the same exposure condition and that $\sigma^1_i=\sigma^0_i$, which corresponds to a treatment that has no effects.

We define an {\it exposure model} for user $i$ as a set of exposure conditions that completely partition the possible assignment vectors $\vec{z}$. The set of all models, across all users, is the exposure model for an experiment.
For our purposes though, it is unnecessary to entirely specify an exposure model, since we are only trying to determine the average treatment effect between the extreme universes. 
We only care about the exposure conditions $\sigma^1_i$ and $\sigma^0_i$ 
for which each user $i$ experiences exposure to the treatment or control universe\footnote{If one was to assume functional relationships between the potential outcomes in different exposure conditions then other exposure conditions besides $\sigma^1_i$ and $\sigma^0_i$ could become relevant.}.

Of course, the {\it true} exposure conditions $\sigma^1_i$ and $\sigma^0_i$ for each user are not known to the experimenter a priori, and analyzing the results of an experiment requires choosing such conditions in our framework. If the wrong exposure conditions are chosen by the experimenter, what happens to the estimate of the average treatment effect? If users are responding in ways that do not correspond to $\vec{z}=\vec 1$ and $\vec{z}=\vec 0$, we will be introducing bias into the average treatment effect. The magnitude of this bias depends on how close the outcomes actually observed are to the outcomes at $\vec{z}=\vec 1$ and $\vec{z}=\vec 0$ that we wanted to observe. It may even be favorable to allow such bias in order to lower variance in the results of the experiment.

\xhdr{Neighborhood Exposure}
We now describe some general exposure conditions that we use in
what follows.
In particular, we focus primarily on {\em local exposure conditions},
where two assignments are indistinguishable to $i$ if they agree
in the immediate graph neighborhood of $i$.
We consider absolute and fractional conditions on the number of treated neighbors.  Note we are not asserting that these possible exposure conditions are the {\it actual} exposure conditions with respect to the actual potential outcomes in an experiment, but rather that they provide useful abstractions for the analysis of an experiment, where again the degree of bias introduced depends on how well the exposure conditions approximate belonging to the counterfactual universes.
\vspace{-0.1cm}
\begin{itemize}[leftmargin=8pt]
\itemsep2pt 
\parskip1pt 
\parsep2pt
\labelsep2pt
\item {\em Full neighborhood exposure:} 
Vertex $i$ experiences full neighborhood exposure to a treatment condition if $i$ and all $i$'s neighbors receive that treatment condition.
\item {\em Absolute $k$-neighborhood exposure:}
Vertex $i$ of degree $d$, where $d \ge k$, experiences absolute $k$-neighborhood exposure to a treatment condition if $i$ and $\ge k$ neighbors of $i$ receive that treatment condition. 
\item {\em Fractional $q$-neighborhood exposure:}
Vertex $i$ of degree $d$ experiences fractional $q$-neighborhood exposure to a treatment condition if $i$ and $\ge \! qd$ neighbors of $i$ receive that treatment condition.
\end{itemize}
\vspace{-0.1cm}
The $k$-absolute and $q$-fractional neighborhood exposures can be considered relaxations of the full neighborhood exposure for vertex $i$ in that 
they require fewer neighbors of $i$ to have a fixed treatment condition for $i$ to be considered as belonging to that exposure condition.  In fact, the set of assignment vectors that correspond to $k$-absolute and $q$-fractional neighborhood exposures are each nested under the parameters $k$ and $q$ respectively.  Increasing $k$ or $q$ decreases the set of assignment vectors until reaching full neighborhood exposure for vertex $i$.

It is natural to consider heterogeneous values $k$ or $q$ --- values that differ for each user -- but
we limit our discussion to exposure conditions that are homogeneous across users as much as possible. We do incorporate a mild heterogeneity in the definition of $k$-neighborhood exposure when vertices have degree $d<k$: for these vertices we consider full neighborhood exposure instead.   Fractional exposure does not require this adjustment.

\xhdr{Core Exposure}
Full neighborhood exposure is clearly only an approximation of full immersion in a universe. 
Beyond local exposure conditions, 
we also consider exposure condition with global dependence. As one approach, consider individuals as exposed to a treatment only if they are sufficiently surrounded by sufficiently many treated neighbors who are in turn also surrounded by sufficiently many treated neighbors, and so on. This recursive definition may initially appear intractable, but such recursive exposure can in fact be characterized precisely 
by analyzing the $k$-core --- and more generally the 
heterogeneous $k$-core --- on the induced graph of treatment 
and control individuals.

Recall that the $k$-core of a graph $G=(V,E)$ is the maximal subgraph of $G$ in which all vertices have degree at least $k$ \cite{Bollobas2001}. Similarly, the heterogeneous $\mathbf{k}$-core of a graph $G=(V,E)$, parameterized by a vector $\mathbf{k} = (k_1,\ldots,k_{|V|})$, is the maximal subgraph $H=(V',E')$ of $G$ in which each vertex $v_i \in V'$ has degree at least $k_i$ \cite{CLDG2012}. Using the definition of heterogeneous $\mathbf{k}$-core, we introduce the following natural fractional analog.
\begin{definition}[Fractional $q$-core]
The fractional $q$-core is the maximal subgraph $H=(V',E')$ of $G=(V,E)$ in which each vertex $v_i \in V'$ is connected to at least a fraction $q$ of the vertices it was connected to in $G$. Thus, for all $v_i \in V'$, $\deg_H(v_i)  \ge q \deg_G(v_i)$. 
Equivalently, if $d_i$ is the degrees of vertex $i$, the fractional $q$-core is the heterogenous $\mathbf{k}$-core of $G$ for $\mathbf{k}=(qd_1,\ldots,qd_{|V|})$.
\end{definition}

Since the heterogeneous $\mathbf{k}$-core is a well-defined object, so is
the fractional $q$-core.
Using this definition, we now define exposure conditions that are all 
stricter versions of corresponding earlier neighborhood conditions.

\vspace{-0.12cm}
\begin{itemize}[leftmargin=8pt]
\itemsep2pt 
\parskip5pt 
\parsep3pt
\labelsep2pt
\item {\em Component exposure:}
Vertex $i$ experiences component exposure to a treatment condition if $i$ and all of the vertices in its connected component receive that treatment condition. 
\item {\em Absolute $k$-core exposure:}
Vertex $i$ with degree $d \ge k$ experiences absolute $k$-core exposure to a treatment condition if $i$ belongs to the 
$k$-core of the graph $G[V']$, the subgraph of $G$ induced on the vertices $V'$ that receive that treatment condition.
\item {\em Fractional $q$-core exposure:}
Vertex $i$ experiences fractional $q$-core exposure to a treatment condition if $i$ belongs to the fractional $q$-core of the graph $G[V']$, the subgraph of $G$ induced on the vertices $V'$ that receive that treatment condition. 
\end{itemize}
\vspace{-0.12cm}

Component exposure is perhaps the strongest requirement for network exposure imaginable, and it is only feasible if the interference graph being studied is comprised of many disconnected components. We include it here specifically to note that the fractional $q$-core exposure for $q=1$ reduces to component exposure.  Again like the neighborhood exposure case, absolute core exposure requires heterogeneity in $k$ across users for it to be a useful condition for all users.  A parsimonious solution analogous to the solution for $k$-neighborhood exposure may be to consider heterogeneous max(degree, $k$)-core exposure. Fractional $q$-core exposure, like fractional $q$-neighborhood exposure, is again free from these parsimony problems.

Core exposure conditions are strictly stronger than the associated neighborhood exposure conditions above.  In fact, every assignment vector in which a vertex $i$ would be component or core exposed corresponds to neighborhood exposure, but not vice versa.  So the assignment vectors of core and component exposure are entirely contained in those of the associated neighborhood exposure.

\xhdr{Other Exposure Conditions} 
Other exposure conditions may prove relevant to particular 
applications. In particular, we draw attention to the intermediate concept of placing absolute or fractional conditions on the population of vertices within $h$ hops, where $h=1$ is the neighborhood exposure conditions above. We also note that on social networks with very high degree, for many applications it may be more relevant to define the exposure conditions in terms of a lower degree network that considers only stronger ties.

\section{Randomization and Estimation}
Using the concept of network exposure, we can now consider estimating the average treatment effect $\tau$ between the two counterfactual universes using a randomized experiment.  Recall that $\vec z$ is the treatment assignment vector of an experiment. To randomize the experiment, let $\vec z$ be drawn from $Z$, a random vector that takes values on $\{0,1\}^n$, the range of $\vec z$. The distribution of $Z$ over $\{0,1\}^n$ given by $\Pr(Z = \vec{z})$ is what defines our randomization scheme, and it is also exactly what determines the relevant probabilities of network exposure. For a user $i$, $\Pr(Z \in \sigma^1_i)$ is the probability of network exposure to treatment and $\Pr(Z \in \sigma^0_i)$ is the probability of network exposure to control.  

In general, these probabilities will be different for each user and each treatment condition, and knowing these probabilities makes it possible to correct for allocation bias during randomization. In particular, it becomes possible to use the Horvitz-Thompson estimator, $\hat{\tau}$, to obtain an unbiased estimate of $\tau$, here given by
\begin{eqnarray}
\hat{\tau}(Z) = \frac{1}{n}{\sum_{i=1}^n \left( \frac{Y_i(Z)\mathbf{1}[Z \in \sigma^1_i]}{\Pr(Z \in \sigma^1_i)} -\frac{Y_i(Z)\mathbf{1}[Z \in \sigma^0_i]}{\Pr(Z \in \sigma^0_i)}\right)},
\end{eqnarray}
where $\mathbf{1}[x]$ is the indicator function. Assuming the probabilities are positive, the expectation over $Z$ clearly gives $\tau$, though note that this does assume that the exposure conditions are not misspecified.

Let us examine the exposure probabilities for the simplest network
exposure condition, full neighborhood exposure, and under the simplest
randomization scheme --- independent vertex randomization, in which 
each vertex is independently assigned to treatment or control. If all
vertices are treated independently with probability $p \in (0,1)$ then the
probability of full neighborhood exposure to treatment for a user $i$
of degree $d_i$ is simply given by $\Pr(Z \in \sigma^{1}_i) = p^{d_i+1}$,
and the probability of full neighborhood exposure to control is given
by $\Pr(Z \in \sigma^{0}_i) = (1-p)^{d_i+1}$. This highlights
the main challenge of network exposure: the chance that a vertex with
high degree manages to reach full neighborhood exposure, or anywhere
near it, can be exponentially small in $d_i$. Intuitively, such small
exposure probabilities will dramatically increase the variance of the
Horvitz-Thompson estimator, and it indicates the necessity of using
more intelligent randomization.

To reduce the variance of this Horvitz-Thompson estimator, 
we introduce a general {\em graph cluster randomization} approach,
creating graph clusters and randomizing assignment
at the cluster level
rather than at the vertex level, with clusters assigned independently.
Connected vertices
will then be assigned to the same treatment condition 
more often than would happen with independent assignment, 
increasing the expected number of users who are network 
exposed to a condition at the cost of increased correlations between users' exposure conditions.

For clarity when discussing clustering, we introduce some notation.
Let the vertices be partitioned into $n_c$ clusters $C_1,\ldots,C_{n_c}$.
Let $N_i \subseteq V$ denote the neighbors of $i$ in the graph $G$, and let
$S_i  = \{ C_j : (i \cup N_i) \cap C_j \neq \emptyset \}$
denote the set of clusters that contain $i$ or a neighbor of $i$;
we call $S_i$ the set of clusters to which $i$ is {\em connected}.
Using this notation, 
we will now examine the probabilities of different network exposures. 

For the general creation of clusters we defer to the literature on 
algorithms for graph partitioning and community detection \cite{Fortunato2010,Ugander2013}. 
In Section \ref{s:growth} we describe a particular algorithm for clustering
graphs that satisfy a restricted-growth condition. The remainder of this section,
however, describes the behavior of an arbitrary clustering on an arbitrary graph.

\subsection{Exposure probabilities}

We now examine how the probabilities of network exposure 
can be computed given a clustering. 
As a simple example, for the full neighborhood exposure condition, 
the probability
of network exposure to treatment simply becomes $\Pr(Z \in \sigma^0_i) = p^{|S_i|}$
 and to control becomes $\Pr(Z \in \sigma^1_i) = (1-p)^{|S_i|}$. 
We now show that computing the exposure probabilities for 
absolute and fractional neighborhood exposure conditions is tractable as well. 

Consider the challenge of computing the probability that vertex $i$ with degree $d_i$ is treated and more than $k$ of its neighboring vertices are treated under cluster randomization. This applies when considering both absolute and fractional neighborhood exposures. First, let us reindex the clusters such that if $i$ is connected to $|S_i| = s$ clusters, $i$ itself resides on cluster $s$, and we let $j=1,\ldots,s-1$ denote the other connected clusters. Let $w_{i1},\ldots,w_{is}$ be the number of connections $i$ has to each cluster, and let the Bernoulli$(p)$ random variables $X_1,\ldots,X_s$ denote the independent coin tosses associated with each cluster. Then:
\begin{align*}
\Pr[ Z \in \sigma^1_i ]
=& \Pr \left [ X_s = 1 \right ] \cdot \Pr \left [ \textstyle{\sum_{j=1}^{s-1}} w_{ij} X_j \ge k - w_{is} \right ], \\
\Pr[ Z \in \sigma^0_i ] 
=& \Pr \left [ X_s = 0 \right ] \cdot \Pr \left [ \textstyle{\sum_{j=1}^{s-1}} w_{ij} X_j \le d_i - k \right ].
\end{align*}
Here the random quantity $\sum_j w_{ij} X_j$ obeys a weighted equivalent of a Poisson-binomial distribution, and the probabilities in question can be computed explicitly using a dynamic program defined by the following recursion
\begin{eqnarray*}
\Pr \left [ \textstyle{\sum_{j=1}^{s}} w_j X_j \ge T \right ] = 
& p \Pr \left [ \textstyle{\sum_{j=1}^{s-1}} w_{ij} X_j \ge T - w_{is} \right ] + \\ 
& (1-p) \Pr \left [ \textstyle{\sum_{j=1}^{s-1}} w_{ij} X_j \ge T \right ].
\end{eqnarray*}
Note that $T$ is bounded by the maximum vertex degree $d_{max}$, 
making this a polynomial time dynamic program with runtime $O(d_{max}s)$. We formalize this computation into the following proposition.
\begin{proposition}
The probability that vertex $i$ is treated and $\ge k$ neighboring vertices are treated under independent cluster randomization is given by 
$$\Pr[Z \in \sigma_i^1] = pf(s{-}1,k{-}w_{is};p,\vec w)$$
 where 
\begin{eqnarray*}
f(1,T;p,\vec w_i) &=& p \mathbf{1}[T< w_{i1}], \\
f(j, T; p,\vec w_i) &=& p f(j-1,T-w_{ij}; p, \vec w_i  ) \\
&& + (1-p) f(j-1,T;p, \vec w_i).
\end{eqnarray*}
The probability that vertex $i$ is in control and $\ge k$ neighboring vertices are in control under independent cluster randomization is given by 
$$\Pr[Z \in \sigma_i^0] = (1-p)[1-f(s - 1,d_i-k+1;p,\vec w)].$$
\end{proposition}
Recall that these partial neighborhood exposure conditions (absolute and fractional) are nested. In fact, for a given vertex $i$ the recursion can be used to derive the probability for every possible threshold value under consideration in a single $O(d_{max}s)$ double for-loop. Such a computation in fact returns the probability distribution over the exposure space for each individual. 
See Figure \ref{f:personal} for illustrations of what this distribution can look like.

\begin{figure}[t]
\begin{center}
\includegraphics[width =.85\columnwidth]{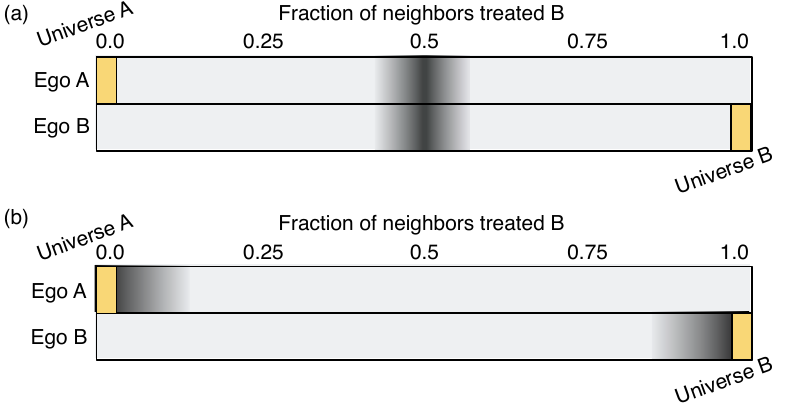} 
\vspace*{-0.20in}
\caption{
The probability distribution over the exposure space for a single individual, where the exposure conditions $\sigma_i^0$ and $\sigma_i^1$ are shown in yellow for both (a) an i.i.d. vertex randomization and (b) an ideal cluster randomization, where the probability mass is collected at exposure conditions of interest.
\vspace*{-0.25in}
}
\label{f:personal}
\end{center}
\end{figure}

The dynamic program above only provides
a means of exactly computing exposure probabilities for absolute and fractional 
neighborhood exposure conditions. 
Unfortunately, how to efficiently compute the exact probability of $k$-core and fractional $q$-core
exposure conditions is unclear, but recall that these exposure conditions were
formally nested subsets of the corresponding neighborhood exposure conditions. 
This at least allows us to upper bound the core exposure probabilities, and we formalize this connection via the following proposition. Because we are generally concerned about exposure probabilities being too small, this upper bound can be useful in identifying vertices with problematically small probabilities already under neighborhood exposure.

\begin{proposition} The probability vertex $i$ is network exposed to a 
treatment condition under core exposure is less than or equal to the 
probability under the analogous neighborhood exposure:  
\begin{eqnarray*}
\Pr( Z \in \sigma_i^x | k\text{-core}) &\le& Pr( Z \in \sigma_i^x | k\text{-nhood}), \\
\Pr( Z \in \sigma_i^x | \text{ frac } q\text{-core}) &\le& Pr( Z \in \sigma_i^x | \text{ frac } q\text{-nhood}).
\end{eqnarray*}
\end{proposition}

It is possible that a useful direct estimate of the core exposure probabilities can be obtained 
via Monte Carlo sampling of the randomization, but we do not explore that possibility here.

\subsection{Estimator variance}

The variance of the Horvitz-Thompson estimator under interference has been studied 
by Aronow and Samii \cite{AS2012}, where they also present 
several variance reduction schemes. 
Estimating the variance under their approach requires 
knowledge of joint exposure 
conditions, the joint probability that vertex $i$ is network exposed to treatment/control 
and vertex $j$ is network exposed to treatment/control. This is the probability that the 
random vector $Z$ is in the exposure condition for vertex $i$ and for vertex $j$ 
simultaneously, i.e. $\Pr(Z \in (\sigma^1_i \cap \sigma^1_j))$ for joint network exposure to treatment.  
If one is interested in computing the variance of the estimator analytically then there is 
nothing fundamentally different about this probability computation when compared 
to the single vertex exposure probability, aside from
the fact that the intersection of 
the two sets can be empty. Aronow and Samii observe that an empty intersection 
makes it impossible to derive an unbiased estimate of the variance (though they
show how the variance can still be upper bounded), but it does 
not bias the effect estimator itself.

The variance of the effect estimator where 
$$\hat{Y}^x(Z)=\frac{1}{n} \sum_i \left[Y_i(Z)\mathbf{1}[Z \in \sigma^{x}_i]/\Pr(Z \in \sigma^{x}_i)\right]$$
 is given by
\begin{eqnarray}
\label{htvar}
\nonumber 
\Var[ \hat{\tau} (Z)]&=  \left [ 
\Var[\hat{Y}^1(Z)] +
\Var[\hat{Y}^0(Z)] \right. - 
\\ & \left. 2 \Cov[\hat{Y}^1(Z), \hat{Y}^0(Z)]
\right ].
\end{eqnarray}
Assuming the exposure conditions are properly specified, namely assuming that $Y_i(\vec z)$ is constant for all $\vec z \in \sigma^x_i$, we can introduce the notation $Y_i(\sigma^x_i) := Y_i(\vec z \in \sigma^x_i)$. Using the further notation $\pi_i^x := Pr[Z \in \sigma^x_i]$ and $\pi_{ij}^{xy} := Pr[Z \in (^x_i \cup \sigma^y_j)]$ we obtain
\begin{eqnarray} \nonumber
\Var [ \hat{Y}^x(Z)]
&=& \frac{1}{n^2} 
 \Bigg [
\sum_{i=1}^{n} \frac{1-\pi_i^x}{\pi_i^x} Y_i(\sigma^x_i)^2 + \\
&& \sum_{i=1}^{n} \sum_{\substack{j=1 \\ j\ne i}}^n \frac{ \pi_{ij}^{xx} - \pi_i^x \pi_j^x}{\pi_i^x \pi_j^x} Y_i(\sigma^x_i) Y_j(\sigma^x_j) \Bigg ]
,
\end{eqnarray}
and
\begin{eqnarray} 
\nonumber 
\Cov[\hat{Y}^1(Z), \hat{Y}^0(Z)]  =
\frac{1}{n^2} \Bigg[ \sum_{i=1}^{n} \sum_{\substack{j=1 \\ j\ne i}}^n  \frac{ \pi_{ij}^{10} - \pi_i^1 \pi_j^0}{\pi_i^1 \pi_j^0} Y_i(\sigma^1_i) Y_j(\sigma^0_j) -  &\\
 \sum_{i=1}^n Y_i(\sigma^1_i)Y_i(^0_i) \Bigg ]. \hspace{1cm}&
\label{eq:covar}
\end{eqnarray}

The above expressions make it evident that the variance is very tightly controlled 
by the probabilities of exposure, and in order to upper bound the variance we will require
lower bounds on the probabilities $\pi_i^x$ and also upper bounds on the
joint probabilities $\pi_{ij}^{xy}$, for all vertex pairs and all combinations of $x$ and $y$. 
For neighborhood exposure, we can now write basic sufficient conditions
under which the variance of the estimator is asymptotically $O(1/n)$ in $n$ for
graph cluster randomization.

\begin{proposition}
Assume the potential outcomes $Y_i(\cdot)$ are all $O(1)$ in $n$. 
If $G$ has maximum degree $O(1)$ and the size of each 
cluster is $O(1)$, then the variance of the Horvitz-Thompson estimator 
for full, $k$-neighborhood, and $q$-fractional neighborhood exposure under 
graph cluster randomization is $O(1/n)$.
\end{proposition}
\begin{proof}
Assume $G$ has maximum degree $O(1)$ and the size of each cluster is $O(1)$.
All of the single sums are clearly $O(n)$: $\pi_i^{x}$ is O$(1)$ since all vertices have bounded degree.
For the double sums, note that $\pi_{ij}^{xx} = \pi_i^{x}\pi_j^{x}$ if and only if $i$ and $j$
have no common cluster neighbors, $|S_i \cap S_j| = 0$.  
Whenever $|S_i \cap S_j| > 0$, $\pi_{ij}^{xx} >\pi_i^{x}\pi_j^{x}$ for full, $k$-neighborhood, 
and $q$-fractional neighborhood exposure. Further, 
$\pi_{ij}^{10} < \pi_i^{1} \pi_j^{0}$ if $|S_i \cap S_j| > 0$ and 
$\pi_{ij}^{10} = \pi_i^{1} \pi_j^{0}$ otherwise.  

So the terms of the double sums are zero whenever $\pi_{ij} = \pi_i \pi_j$
and when the terms are not zero ($|S_i \cap S_j| > 0$), 
they are all positive and bounded above $O(1)$ due to the bounded degrees.  We now bound
the number of vertices $j$ for which $|S_i \cap S_j| > 0$.  Vertex $i$ at most connects to $O(1)$ clusters and therefore $|S_i| = O(1)$.  For all $C \in S_i$, we have that $|S_i \cap S_j| > 0$ for any $j \in C$ and for any vertex $j$ that is adjacent to a vertex in cluster $C$.  Both of these contributions is $O(1)$, giving an $O(1)$ contribution of vertices for each $C \in S_i$.  Since there are $O(1)$ such clusters, this is still $O(1)$ vertices $j$ in total for vertex $i$ such that $|S_i \cap S_j| > 0$.  Thus for each vertex, at most $O(1)$ of the terms in the double sum are positive, making the total variance $O(1/n)$.
\end{proof}

The strength of this general result is that it achieves an 
$O(1/n)$ bound on the variance when the maximum degree is bounded.
The problem is that the variance can grow exponentially 
in the degrees of the graph.
In this next section we address this issue,
introducing a condition on a graph that ensures we can 
find a clustering into sets of size $O(1)$ ---
consistent with the above result -- for which the variance 
grows as $O(1/n)$ but is also linear rather than exponential
in the maximum degree.

\begin{figure*}[t]
\begin{center}
\includegraphics[width =.75\textwidth]{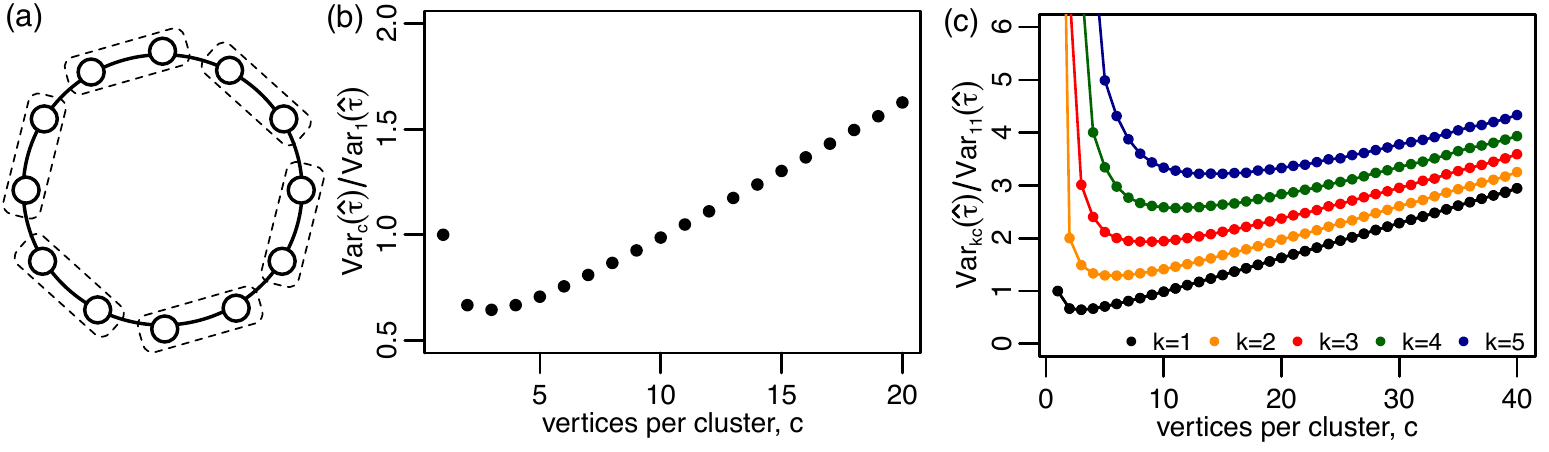} 
\vspace*{-0.15in}
\caption{
The cycle graph, (a) where vertices respond $\bar Y$ to treatment and $0$ to control, shown clustered in groups of $c=2$ vertices. 
(b) Asymptotic variance of the estimator for this graph as a function of the number of vertices per cluster, normalized by estimator variance for $c=1$ vertices per cluster. (c) Simulated variance of the estimator for $k$th powers of the cycle graph for $k=1,\ldots,5$ as a function of the number of vertices per cluster. For each $k$ the variance for cluster size $c=2k+1$ grows linearly in $k$.
\vspace*{-0.30in}
}
\label{f:toymodel}
\end{center}
\end{figure*}

\section{Variance on \\ Restricted-growth graphs}
\label{s:growth}

In order to measure average treatment effects under interference 
on large-scale graphs, it is necessary to design
a randomization scheme capable of containing the estimator 
variance for high-degree vertices. 
In this section we show that any graph satisfying our 
restricted-growth condition admits a
clustering that can produce an unbiased effect estimate that is both $O(1/n)$
and linear in the degrees of the graph. 
In contrast, we show that
with less careful clustering,
it is easy for 
the variance to grow exponentially in the degrees.

Let us first define restricted-growth graphs. Let $B_{r}(v)$ be the set of vertices within $r$ hops of a vertex $v$. 

\begin{definition}
A graph $G=(V,E)$ is a restricted-growth graph if for all vertices $v \in V$ and all $r > 0$,
we have $|B_{r+1}(v)| \leq \kappa |B_r(v)|$. 
\end{definition}

As mentioned in the introduction, 
graphs derived from a
uniform-density embedding in a Euclidean space of dimension $m$ exhibit restricted growth,
with growth constant $\kappa = 2^m$ independent
of degree.
To develop intuition for the restricted-growth assumption, 
we first analyze the
variance using graph cluster randomization 
on a family of particularly 
tractable restricted-growth graphs, $k$th powers of the cycle. 
We follow this analysis by
proving bounds on the variance for general restricted-growth graphs.

\subsection{Cycle and powers of the cycle examples}

First we will consider a simple  graph consisting of a single cycle with $n$ vertices. 
For this graph, we consider the full neighborhood exposure model, where we are interested 
in the average treatment effect between $\sigma^1_i$, 
when a vertex is treated and both of their neighbors are treated, 
and $\sigma^0_i$, when a vertex is not treated and neither of their neighbors are treated.
For the fixed responses of the vertices to treatment and control, we assume that
all vertices uniformly respond $Y_i(\sigma^1_i) = \bar Y$ to network exposure to the treatment
and $Y_i(\sigma^0_i)=0$ to network exposure to the control. 
The cycle graph clearly admits an intuitively obvious
clustering using the cycle structure, 
with contiguous blocks
of $c$ vertices randomized together. 
As a last assumption, assume that clusters are selected
under a balanced randomization with $p=1/2$. 
Our goal is to determine how the variance of the 
Horvitz-Thompson average treatment effect estimator depends on the size $c$ of these clusters.
For this basic combination of graph, exposure condition, responses, and clustering,
one can derive the asymptotic variance exactly.

Consider the variance presented in \eqref{htvar} above.
Since all vertices respond zero to the control condition in our example, 
as long as the exposure probability for the control condition is strictly positive then 
both $\Var(\hat{Y}(\sigma_0))$ and $\Cov(\hat{Y}(\sigma_1), \hat{Y}(\sigma_0))$ are zero. 
Since our calculations will rely only on probabilities $\pi^1_i$ for the
exposure to treatment condition, we omit the superscript. The variance is then:
\begin{eqnarray} 
\label{examplevar}
\Var [ \hat{\tau} (Z)]
= \frac{\bar Y^2}{n^2} 
 \Bigg [
\sum_{i=1}^{n}\left  (\frac{1}{\pi_i}-1 \right) + 
\sum_{i=1}^{n} \sum_{\substack{j=1 \\ j\ne i}}^n \left ( \frac{ \pi_{ij}}{\pi_i \pi_j} -1 \right )\Bigg ] .
\end{eqnarray}
Notice that the terms of the double sum are only non-zero for vertex pairs where $\pi_{ij} \ne \pi_i \pi_j$.

First, consider the case of each vertex being its own cluster. The probability of being exposed and both of 
one's neighbors being exposed is equal to the probability of seeing three independent 
coins come up heads. When the randomization is balanced (e.g. $p=1/2$), we 
obtain $\pi_{i}=1/8, \forall i$.
Note that the co-assignment probabilities depend on whether vertices $i$ and $j$ 
are neighbors or share a neighbor. From this we derive $\pi_{ij} = 1/16$ if $|i-j|=1$ 
and $\pi_{ij} = 1/32$ if $|i-j|=2$, and if $|i-j|>2$, the probabilities are independent.
We obtain $\Var( \hat{\tau} (Z)) =(15/2)\bar Y^2\frac{1}{n} + O(1/n^2)$.

%
%

Now, consider randomizing blocks of $c\ge 2$ vertices, where $c$ does not depend on $n$.
The calculations for this case are expansive but straight-forward.
We consider a single one of the equivalent cyclically shifted possibilities. 
The calculation requires handling $c=2$ and $c \ge 3$ separately, but 
the expression for $c \ge 3$ as a function of $c$ holds for $c=2$ as well, so
we omit the special case for brevity.
The variance calculation depends on distance $\Delta = |i-j|$ up to $\Delta =c+1$, and 
for $c\ge3$ this evaluates to:
\begin{align*}
\Var[ \hat{\tau} (Z)] 
= \frac{\bar Y^2}{n^2}
\Bigg [
\left( n +\frac{4n}{c} \right)
+ 
\underbrace{
\frac{2n}{c}
(c+2)
}_{\Delta=1}
+ \hspace{2cm}
\\ 
\underbrace{
\frac{2n}{c} \sum_{k=2}^{c-2} (c - k + 2 )
}_{1<\Delta<c-1}
+ 
\underbrace{
\frac{2n}{c} 3 
}_{\Delta=c-1}
+ 
\underbrace{
\frac{2n}{c} 2
}_{\Delta=c}
+ 
\underbrace{
\frac{2n}{c}
}_{\Delta=c+1}
\Bigg ] + O\left (\frac{1}{n^2} \right ).
\end{align*}
This reduces to 
$\Var( \hat{\tau} (Z)) =
\left (\frac{c}{2} + 2 + \frac{4}{c} \right ) \bar Y^2 \frac{1}{n} + O(1/n^2)$,
which holds for all $c \ge 2$.

Combining these calculations, the asymptotic variance of the estimator for all $c$ 
is plotted in Figure 2. Notice that the variance is minimized when randomizing 
clusters of size $c=3$, which corresponds exactly to the size of neighborhoods
on the simple cycle. 

To build upon this observation, we now examine the simulated variance for higher degree
extensions of the cycle, the so-called $k$th power of the cycle, where analytic derivation is already unwieldy.
Thus, we use a simulation of 
the cluster randomization procedure to examine how the variance of the 
effect size estimator depends on the cluster size for these higher degree graphs.

The $k$th power of a cycle graph consists of a cycle where each vertex is connected
to the $k$ nearest neighbors on each side, yielding a regular graph where all vertices
have degree $d=2k$. By sampling one million cluster randomizations on graphs
with $n=5000$ vertices, we can compute the sample variance of the estimator across these samples.
The results are shown in Figure 2, for $k=1$ through $k=5$. The simulations for $k=1$ agree
precisely with the overlaid asymptotic calculations.

Notice how the optimal cluster size $c$ appears to scale approximately
linearly in degree, and also notice how the variance at the optimal clustering size,
the minimum value of each curve as $k$ increases, appears to scale linearly in $k$.
While the exact variance as a function of cluster size $c$ is unwieldy to derive, 
we are able to provide the following upper bound, showing how the variance 
of the estimator for clusters of size $c= d+1$
scales linearly in the degree $d$ of the graph.
This suggests that one should treat contiguous blocks 
of the cycle attuned to the size of the neighborhood of the vertices.

When deriving this upper bound, it is no longer necessary to assume a uniform response
$Y_i(\sigma^1_i) = \bar Y$, and instead we simply assume that the responses are upper bounded
by some value $Y_i(\sigma^1_i) \le Y_M$. 

When clusters have size $c=d+1$, each vertex can be connected to at most $2$ clusters,
meaning that $ 1/\pi_i \le 1/p^2$ for all $i$. So
\begin{eqnarray*} 
\label{powervar1}
\Var [ \hat{\tau} (Z)]
&\le& 
\frac{Y_M^2}{n^2}
 \Bigg [
\sum_{i=1}^{n} (p^{-2}-1) + 
\sum_{i=1}^{n} \sum_{\substack{j=1 \\ j\ne i}}^n (\frac{ \pi_{ij}}{\pi_i \pi_j} - 1) \Bigg ].
\end{eqnarray*}
Now each vertex
has a non-independent joint assignment probability 
(such that $\pi_{ij} \neq \pi_{i} \pi_{j}$) with at most $3d+1$ other 
vertices: 
up to $2d+1$ other vertices when they are adjacent to two clusters, 
the $d/2$ to the left of the left cluster, and 
the $d/2$ to the right of the right cluster.
The joint assignment probability $\pi_{ij}$ is at most $p^2$, since
two vertices can not both be at the center of a cluster.
For each $i$, the sum indexed by $j$ then can be bounded, producing
\begin{eqnarray*} 
\label{powervar2}
\Var [ \hat{\tau} (Z)]
\le
Y_M^2 (p^{-2} - 1)(3d+2)\frac{1}{n} .
\end{eqnarray*}

This result tells us that it is possible to experimentally measure network effects 
on a cycle graph of very high degree $d$ with a variance that is only
linear in $d$, provided that the vertices are clustered in contiguous 
blocks of $d+1$ vertices.
We now show how this strategy of bounding the variance 
applies to a much more general class of graphs, using a clustering
algorithm that does not require knowledge of any geometric structure.

\subsection{Clustering restricted-growth graphs}

We now begin developing the main result of this section, a cluster
randomization scheme for the class of restricted-growth graphs.
The first component is a clustering algorithm for such
graphs in which each vertex 
is connected to at most a constant number of clusters, 
independent of the degree of the vertex. This 
will then imply that the variance 
on any restricted-growth 
graph can be upper bounded by a function linear in the degree.
Our clustering shows that the nice decomposition
of the cycle by contiguous regions can be generalized to arbitrary
graphs in our class. 
In other words, 
the geometry isn't crucial; the restricted-growth property is enough.

Consider a restricted-growth graph $G = (V,E)$;
we will present the case in which 
$G$ is $d$-regular, but as we note below, the 
regularity can be relaxed to
arbitrary degree distributions at the cost of a weaker
but still constant bound on the number of connected clusters.

Recall that 
the restricted-growth condition says there exists 
$\kappa$ so that for all $v$
and all $r > 0$, we have $|B_{r+1}(v)| \leq \kappa |B_r(v)|$.
Importantly, $r = 0$ is different:
$B_0(v)$ is the singleton set $\{v\}$, while $B_1(v)$ is
the neighborhood of $v$ and hence has size $d+1$.
Thus $|B_1(v)|/|B_0(v)| = d+1$, potentially much larger than 
the bound of $\kappa$ on the ratio $|B_{r+1}(v)|/|B_r(v)|$ for $r > 0$.
This is the crux of the restricted-growth condition:
from radius 0 to 1 we have 
unrestricted growth (a factor of $d+1$), but then the growth slows
to factors of $\kappa$ which can be bounded separately from $d$.

In the language of metric spaces, we will cluster the graph using a 
{\it 3-net} for the shortest-path metric of $G$ \cite{KL2002}. 
Formally, in a metric space $X$, an $r$-net $Y \subseteq X$ is a collection of
points that are mutually at distance at least $r$ from each other, 
but the union of all their $r$-balls covers the space, $X \subseteq \cup_{y \in Y} B_r(y)$.
Accordingly, we call our construction a 
{\it 3-net clustering} of the graph. 
To build a 3-net clustering, we will iteratively identify vertices $v_1,
v_2, ...$, `marking' vertices as we do this.
Afterwards we will identify clusters $C_1, C_2, ...$ to go with these vertices.
More explicitly, we perform the following procedure consisting of two principle stages:

\vspace{-0.1cm}
\begin{itemize}[leftmargin=11pt]
\itemsep2pt 
\parskip1pt 
\parsep2pt
\labelsep4pt

\item Initially all vertices are unmarked. 

\item  While there are unmarked vertices, in step $j$ find an arbitrary
unmarked vertex $v$, selecting $v$ to be vertex $v_j$ and marking all vertices
in $B_2(v_j)$. 

\item Suppose $k$ such vertices are defined, and let $S = \{v_1, v_2, ..., v_k\}$.

\item For every vertex $w$ of $G$, assign $w$ to the closest vertex $v_i \in S$,
breaking ties consistently (e.g. in order of lowest index).

\item For every $v_j$, let $C_j$ be the set of all vertices assigned to $v_j$.

\end{itemize}

The sets $C_1, \ldots, C_k$ are then our 3-net clustering. 
The key property of this clustering is the following result, 
which establishes that each vertex
is connected to a number of clusters that can be bounded by
a function of $\kappa$, independent of the degree.

\begin{proposition}
\label{claim2}
Consider any 3-net clustering of a graph $G=(V,E)$. For all $w \in V$, the neighborhood $B_1(w)$ has a non-empty
intersection with at most $\kappa^3$ distinct clusters.
\end{proposition}
\begin{proof}
We first claim that for all $v_j \in S$, we have
$C_j \subseteq B_2(v_j)$.
Indeed, consider any vertex $w \neq v_j$ in $C_j$.
We have $w \not\in S$, since otherwise $w$ would belong to the cluster
identified with itself. Now, consider the iteration $i$ in which $w$ was 
marked; we have $w \in B_2(v_i)$.
Since $w \in C_j$ and it is assigned to the closest vertex in $S$,
it follows that $w \in B_2(v_j)$.
Thus $C_j \subseteq B_2(v_j)$.

Next, we claim that for all $v_i, v_j \in S$, 
the sets $B_1(v_i)$ and $B_1(v_j)$ are
disjoint.
Suppose by way of contradiction that 
$B_1(v_i) \cap B_1(v_j) \neq \emptyset$.
It would follow that $v_i \in B_2(v_j)$ and vice versa.
But then if we consider the vertex among $v_i$ and $v_j$ that was added to $S$ first,
the other of $v_i$ or $v_j$ would have been marked in that iteration, 
and hence it could not have been added to $S$ as well.
This contradiction establishes that $B_1(v_i)$ and $B_1(v_j)$ are disjoint.

To complete the proof, suppose by way of contradiction
that $B_1(w)$ has a non-empty intersection 
with more than $\kappa^3$ distinct clusters: 
for some $t > \kappa^3$,
let $u_1, u_2, \ldots, u_t$ be distinct vertices in $B_1(w)$
and $v_{i_1}, \ldots, v_{i_t}$ be distinct vertices in $S$
such that $u_h \in C_{i_h}$ for $h = 1, 2, \ldots, t$.

Since $C_{i_h} \subseteq B_2(v_{i_h})$, and $C_{i_h}$ contains
a vertex adjacent to $w$ (or contains $w$ itself), 
we have $v_{i_h} \in B_3(w)$, and hence
$B_1(v_{i_h}) \subseteq B_4(w)$.
The neighborhoods $B_1(v_{i_1}), B_1(v_{i_2}), \ldots, B_1(v_{i_t})$
are all pairwise disjoint as argued above, and they are all contained in 
$B_4(w)$, which implies that $|B_4(w)| \geq t(d + 1) > \kappa^3(d + 1)$.
But applying the bounded growth inequality $|B_{r+1}(w)| \leq \kappa |B_r(w)|$
three times we have
$|B_4(w)| \leq \kappa^3 (d + 1)$, a contradiction.
This establishes that $B_1(w)$ can have a non-empty intersection 
with at most $\kappa^3$ distinct clusters.
\end{proof}

The above result is formulated for $d$-regular graphs.
But in fact one can show a weaker bound depending only on $\kappa$
as in Proposition \ref{claim2} even for arbitrary restricted-growth graphs,
without any requirement on the degrees.
This weaker bound of $\kappa^6$ can be established by observing
that any restricted-growth graph exhibits a ``bounded gradient''
on the vertex degrees, 
whereby vertices that are near each other in the graph
must have similar degrees.
Combining this fact with 
proof of Proposition \ref{claim2} leads to the desired bound.

\subsection{Variance bounds}
We now apply the above results to bound the variance of the effect estimator $\hat \tau$.
Throughout this section we assume that all responses obey upper bounds and 
positive lower bounds, $Y_i^x \in [Y_m,Y_M]$ for both 
exposure to treatment and control, $x=0,1$. 
The reason for the positive lower bounds is that without them
the users could all be responding zero to all treatments, making the variance zero
regardless of the treatment scheme. We also assume the randomization probability $p$ is 
not degenerate, i.e. $p \in (0,1)$. 
We present the results for $d$-regular graphs to keep expressions
manageable, but analogous results can be derived for arbitrary degrees.

We first establish an exponential lower bound for the variance
under vertex-level randomization, and then we show a contrasting 
linear upper bound for the variance under our 3-net
cluster randomization scheme.

\begin{proposition}
The variance of the HT estimator under full neighborhood exposure 
for vertex randomization
of a graph with $n$ vertices
is lower bounded by an exponential function in the degree $d$ of the graph,
 $\Var [ \hat{\tau} (Z)] \ge O(1/n) (p^{-(d+1)}+(1-p)^{-(d+1)}-1)$.
\end{proposition}
\hspace{3mm} {\sc Proof.}
The joint assignment probabilities for two vertices having the same exposure
is at least the product of their individual probabilities, 
 $\pi^{xx}_{ij} \ge \pi^x_i \pi^x_j$ for $x=0,1$. Thus the double sum in equation (4)
 is non-negative. Similarly, for opposing exposure conditions, we have 
 $\pi^{xy}_{ij} \le \pi^x_i \pi^y_j$ for $x \ne y$, which makes equation 
 (5) a non-negative contribution to equation (3). 
 We focus our lower bound on the main term of equation (4). 
Inputting the probabilities $\pi^1_i = p^{d+1}$ and $\pi^0_i = (1-p)^{d+1}$ 
and lower bounding responses gives us
the desired result. 
\begin{eqnarray*}
\Var [ \hat{\tau} (Z) ]
&\ge& 
\frac{1}{n^2} 
 \Bigg [
\sum_{i=1}^{n} (\frac{1}{\pi^1_i}-1)(Y^1_i)^2
+
\sum_{i=1}^{n} (\frac{1}{\pi^0_i}-1)(Y^0_i)^2 
\Bigg ]
\\
&\ge&
\frac{ Y_m^2}{n} 
 (p^{-(d+1)}+(1-p)^{-(d+1)}-2 ). ~~~~~~~ \Box
\end{eqnarray*}

For graphs with arbitrary degree distributions, this bound becomes
 $
\Var [ \hat{\tau} (Z)] \ge O(1/n) \sum_{i=1}^n (p^{-(d_i+1)}+(1-p)^{-(d_i+1)}-2)
 $, which is exponential in the degree of each vertex, meaning
 that even a single high degree vertices can easily
 explode the variance. 
 
 We now turn to our linear upper bound for
 growth-restricted graphs when using our 3-net clustering.

\begin{proposition}
The variance of the HT estimator under full, $q$-fractional, or $k$-absolute
neighborhood exposure for a 3-net cluster randomization 
of a restricted-growth graph 
is upper bounded by a function linear in the degree $d$ of the graph.
\end{proposition}
\begin{proof}
Recall that the variance of the estimator is given by:
$
\Var( \hat{\tau} (Z))= 
\Var(\hat{Y}^1) +
\Var(\hat{Y}^0) - 
 2 \Cov(\hat{Y}^1, \hat{Y}^0)
.
$
We begin by upper bounding the variance of $\hat{Y^1} (Z)$, and the upper bound
for $\hat{Y^0} (Z)$ follows the same principle. We conclude by bounding the covariance 
term.
By Proposition \ref{claim2},
each vertex is connected to at most $\kappa^3$ clusters. Thus we have the
lower bound $\pi^1_i \ge p^{\kappa^3}$, for both full and fractional 
neighborhood exposure.
\begin{eqnarray*}
\Var [ \hat{Y^1} (Z) ]
\le
\frac{Y_M^2}{n^2} 
 \Bigg [
n(\frac{1}{p^{\kappa^3}}-1) + 
\sum_{i=1}^{n} \sum_{\substack{j=1 \\ j\ne i}}^n (\frac{ \pi^1_{ij}}{\pi^1_i \pi^1_j} - 1) \Bigg ].
\end{eqnarray*}
For each vertex $i$, the inner of the two sums is only nonzero at those vertices $j$ for which the 
assignments are dependent. 
If the assignments for $i$ and $j$ are dependent, then they must
each have neighbors in the same cluster $C_h$ associated with a vertex
$v_h$ in the set of cluster centers.  
Since the proof of Proposition \ref{claim2}
established that $C_h \subseteq B_2(v_h)$, it follows that 
$i$ and $j$ are each within distance 3 of $v_h$ and hence within
distance 6 of each other.
Thus, any $j$ whose assignment is dependent on $i$'s must lie
within $B_6(i)$, and so by the restricted-growth condition,
there can be at most $|B_6(i)| \leq \kappa^5 |B_1(i)| = \kappa^5 (d+1)$
such vertices $j$.
Thus the sum over such $j$ has at most $\kappa^5 (d+1)$ terms.
Also, $\pi^1_{ij} \le p$ applies, 
since the two vertices must depend on at least one cluster.
We obtain
\begin{eqnarray*}
\Var [ \hat{Y^1} (Z) ]
\le
Y_M^2 
 [
(p^{-\kappa^3}-1) + 
\kappa^5 (d+1) (p^{-2\kappa^3-1} - 1) 
]\frac{1}{n}.
\end{eqnarray*}

Now, consider the contribution of the covariance term to the variance, $-2\Cov(\hat Y^1,\hat Y^0)$,
a positive quantity. 
Starting from equation (\ref{eq:covar}), we apply the upper bound for the responses
$Y_i$ to obtain
\begin{eqnarray*}
- 2 \Cov[\hat Y^1(Z),\hat Y^0(Z)] 
\le
- \frac{2Y_M^2}{n^2} 
\sum_{i=1}^{n} \sum_{\substack{j=1 \\ j\ne i}}^n \left ( \frac{\pi_{ij}^{10}}{\pi_i^1 \pi_j^0} - 1 \right )
+ \frac{2Y_M^2}{n}.
\end{eqnarray*}
As with the previous analogous expression, for each $i$ the inner sum is non-zero for at most $\kappa^5(d+1)$ 
other vertices $j$. For the remaining terms, the quantity $-(\pi_{ij}^{10}/(\pi^1_i \pi^0_j) - 1)$ is trivially upper bounded by $1$. 
Thus we obtain
\begin{eqnarray*}
-2\Cov[\hat Y^1(Z),\hat Y^0(Z)] 
\le
\frac{2Y_M^2}{n} 
[\kappa^5(d+1) + 1].
\end{eqnarray*}
Combining the upper bounds, we obtain a total upper bound that is linear in degree, as desired. 
\end{proof}

The restricted-growth condition we used was derived for regular graphs, but as we noted earlier, 
for restricted-growth graphs with arbitrary degree distributions we can apply a weaker but still constant
bound on the cluster dependencies to obtain a variance bound that is still linear in the degree.

\section{Conclusion}

The design of online experiments is a topic with many open directions
(see e.g. \cite{kohavi2012}); in this work we have
focused on the open question of A/B testing when treatment effects
can spill over along the links of an underlying social network.
We introduced a basic framework for reasoning about
this issue, as well as an algorithmic approach --- graph cluster
randomization --- for designing A/B randomizations of a population
when network spillover effects are anticipated.
Appropriate clustering can lead
to reductions in variance that are exponential in the vertex degrees.
We emphasize that beyond the class of graphs where we prove bounds,
graph cluster randomization is a technique that can be applied
to arbitrary graphs using arbitrary community detection or graph partitioning algorithms, 
though we do not provide any variance bound guarantees for these scenarios.

There are many further directions for research suggested by 
the framework developed here.
A first direction is to formulate a
computationally tractable objective function
for minimizing the variance of the Horvitz-Thompson estimator. 
One approach would be via
minimizing an adversarial variance, as in \cite{KLS2012}.
Another problem that may be relevant is to find
a clustering that minimizes A/A variance for full
neighborhood exposure under the assumption of known control potential outcomes. 
Can good clusterings for A/A variance
lead to good solutions for A/B testing?
We note that A/A variance minimization would not be useful 
when the treatment is expected to be dominated by heterogeneous responses.

Adding further structure to the potential treatment responses
is another interesting direction.
We currently have a discrete notion of
network exposure to treatment and control, but one could ask about 
responses that depend continuously on the {\em extent} of exposure.
As one simple example, we could consider a response that was linear in $k$,
when a vertex had $k$ exposed neighbors.
How could we properly take advantage of such structure to get 
better estimates? 
Methods for analyzing bias under network exposure condition 
misspecification would also be a natural addition to the framework.


\end{document}